\documentclass[11pt]{article}
\usepackage{latexsym}
\usepackage{amsmath,amssymb,amsthm,mathrsfs}
\usepackage{bbm}
\usepackage{graphicx}
\usepackage{epsfig}
\usepackage[right=1in, top=1in, bottom=1.2in, left=1in]{geometry}
\usepackage{setspace}
\usepackage[linewidth=1pt]{mdframed}
\usepackage{lipsum}
\usepackage{float}
\usepackage{hyperref}
\usepackage{tikz}
\usepackage{pict2e,xfp}
\spacing{1.06}
\usepackage{lmodern}
\usepackage{braket}
\usepackage{caption}
\usepackage{subcaption}
\usepackage{amsthm}
\usepackage{hyperref}
\usepackage[capitalize]{cleveref}
\usepackage{authblk}

\newtheorem{definition}{Definition}
\newtheorem{theorem}{Theorem}

\def\mathbi#1{\textbf{\em #1}}

\begin{document}

\title{From String Detection to Orthogonal Vector Problem}
\author[1]{Yunhao Wang}
\author[1]{Tianyuan Zheng}
\author[1,2]{Lior Horesh}
\affil[1]{Department of Computer Science, Columbia University, New York City, NY}
\affil[2]{
IBM T.J. Watson Research Center, Yorktown Heights, NY}
\date{}

\maketitle

\begin{abstract}
Considering Grover's Search Algorithm (GSA) with the standard diffuser stage applied, we revisit the $3$-qubit unique String Detection Problem (SDP) and extend the algorithm to $4$-qubit SDP with multiple winners. We then investigate unstructured search problems with non-uniform distributions and define the Orthogonal Vector Problem (OVP) under quantum settings. Although no numerically stable results is reached under the original GSA framework, we provide intuition behind our implementation and further observations on OVP. We further perform a special case analysis under the modified GSA framework which aims to stabilize the final measurement under arbitrary initial distribution. Based on the result of the analysis, we generalize the initial condition under which neither the original framework nor the modification works. Instead of utilizing GSA, we also propose a short-depth circuit that can calculate the orthogonal pair for a given vector represented as a binary string with constant runtime.
\end{abstract}

\section{Introduction}
Search problems have long been widely used in computational area \cite{praditwong2010software, hogg1996phase}. In the classical setting, without any prior knowledge about the testing data, unstructured search problems have a runtime of at least $O(N)$, which is linear to the size of the datasets. Later on, Grover \cite{grover1996fast} proposed a search algorithm and brought down the runtime of analogous problems to $O(\sqrt{N})$ with the power of quantum computing. It was then proved to be optimal \cite{zalka1999grover}. 

Similarly, solving the Orthogonal Vector Problem (OVP) in an unstructured dataset through classical algorithms is polynomial in the number of different vectors and the dimension of each vector. In fact, Orthogonal Vector Conjecture (OVC) \cite{williams2005new} states that this problem cannot be solved in $O(n^{2-\epsilon})$ with $d = c\log n$, for $n$ vectors in $d$ dimensions, with some $c \geq 1$ and all $\epsilon > 0$. Furthermore, the Strong Exponential Time Hypothesis (SETH) \cite{impagliazzo2001complexity} implies OVC, which means that if we can find a faster algorithm for OVP, we could possibly boost the runtime for solving $k$-SAT \cite{williams2005new}. We suggest that OVP is worth being studied under the quantum computing paradigm, and it seems intuitive that we might be able to reformulate OVP as a search problem so that we can leverage the GSA machinery to it.

Therefore, in this paper, we first provide a brief overview of GSA in section \ref{GSA}. Two different search problems are chosen. The String Detection Problem is introduced in section \ref{sdp} with a uniform initial distribution. The Orthogonal Vector Problem is re-defined under quantum settings with a simplified experiment setup in section \ref{ovp}. We provide detailed implementations and discussions in section \ref{result} including case analysis under modified GSA framework. We also present another perspective for solving OVP in section \ref{extend} which achieves a better result than the one under the GSA framework. The conclusion is given in section \ref{conclusion}.

\section{Grover's Search Algorithm}\label{GSA}

\begin{definition}[Unstructured Search Problem]\label{def:search_prob}
Given a boolean function $f_w(x): \{0,1\}^n \rightarrow \{0,1\}$, where $x$ is a search term and $w$ is the winner, find $w$ such that $f_w(x^* = w) = 1 $. This is an unstructured data search problem for $x \in \{0, 1\}^{n}$. $n$ denotes the bit size of the search space.
\end{definition}

If we want to identify the secret parameter $w$ such that $f_w(x^* = w) = 1$ by trivially query the oracle, it would require $O(N)$ calls to the oracle function, with $N = 2^n$. By phase flip and  amplitude amplification, GSA can measure $w$ with high probability using only $O(\sqrt{N})$ calls on the oracle. 

In the algorithm, Hadamard gates $\mathbi{H}$ are first applied to all qubits so that they are all under the  superposition state $\frac{1}{\sqrt{2}} \left(\ket{0} + \ket{1}\right)$. The oracle is given by the unitary operator $U_w : U_{w} |q\rangle = (-1)^{f(q)} |q\rangle$. Following that transformation, the phase of all winner states $\ket{w}$ should be flipped, and the mean distribution goes down. A diffuser stage is then applied to all qubits. This can also be considered as an amplitude amplification achieved by flipping all states around the updated mean. Therefore, after repeating the above steps for $O(\sqrt{N})$ times, the initial state will rotate towards the winner states $\ket{w}$ gradually and thus $\ket{w}$ can be measured with a higher probability.

\section{Problem Statement}
\subsection{String Detection Problem}\label{sdp}
In the SDP, we use $n$ qubits to represent a binary string $\in \{0,1\}^n$, with each qubit representing one classical bit in a string. The size of the search space is $O(2^n) = O(N)$. Assume that we have prior knowledge on a  winner $\ket{w}$ so that the oracle function $f$ is static. This also suggests that for different winners, the oracle implementation will be different. Formally, the SDP with a single winner is defined as follows:

\begin{definition}[String Detection Problem]\label{def:sdp}
Given $n$ qubits and winner state $\ket{w}$ representing a unique string $\{0,1\}^n$, find the quantum circuit $\mathbi{QC}$ such that given $n$ qubits as input, the top measurement of the output is the state $\ket{w}$.

\end{definition}

From the definition, we see that SDP can be treated as an unstructured search problem. Meanwhile, SDP with multiple winners can be easily extended based on this setting by analyzing the final measurement against the superposition of all winner states.

\subsection{Orthogonal Vector Problem}\label{ovp}
OVP aims to identify vectors in an unstructured sequence that are orthogonal to the reference vector. The classical versions of OVP find the orthogonal counterparts for the reference vector by traversing through all candidates and checking the inner products. Thus, the time cost should be linear to the vector dimension and the size of the dataset. However, in quantum circumstances, it is presumable to obtain an implementation quadratically faster by applying the essence of GSA. 

Unlike SDP with GSA, the oracle in OVP is not static and needs to verify the orthogonality between vector pairs instead of checking the identity. We start with the most basic case of vectors $\{0,1\}^2$. The formal problem is defined as follows:

\begin{definition}[Basic Orthogonal Vector Problem]\label{def:ovp_1}
Given a binary reference vector $v_r$ $\in \{0, 1\}^2$ and two other testing vectors $v_1, v_2 \in \{0, 1\}^2$, we initialize the reference qubit and two other testing qubits to their corresponding superposition states $\ket{v_r}, \ket{v_1}, \ket{v_2}$. Find $v_{i, i\in \{1,2\}}$ in the orthogonal vector pair $(v_r, v_i)$, s.t., $\langle v_r, v_i\rangle = 0$.
\end{definition}

In order to identify the indices $i$, we view the final measurement of quantum circuit as a quantum-representation of the testing qubit that is orthogonal to the reference qubit. For example, if the reference qubit is on a state with vector representation $\left(\begin{array}{cc}
    0 & 1
\end{array}\right)^T$, and the two testing qubits are on $\left(\begin{array}{cc}
    1 & 0
\end{array}\right)^T$, $\left(\begin{array}{cc}
    0 & 1
\end{array}\right)^T$ respectively. The expected top measurement should be $\ket{10}$, indicating that the first qubit is orthogonal to the reference qubit.

\section{Implementation and Results}\label{result}
\subsection{String Detection Problem}
Multiple works have been performed on no more than $3$-qubit search problems with single winners \cite{figgatt2017complete, kwiat2000grover, jones1998implementation}. We extend the previous works and focus on a $4$-qubit circuit with multiple winners. Consider all $2^4 = 16$ possible states, we denote each state as a binary string $\{0,1\}^4$, so that we have our winners $w\in \{0,1\}^4$.

Suppose that the winner is the string $\mathbf{0110}$. If we directly apply the oracle strategy for $3$-qubit on $4$-qubit, the results would be as shown in Fig.\ref{fig:invalid_4}. Notice that although all detected "winners" have $\mathbf{11}$ to be the value of the middle two qubits, there are no constraints on the other two qubits. Hence, all combinations $\in \{0,1\}^2$ of the first and the last qubit show up in the final measurement.

\begin{figure}[h]
     \centering
     \begin{subfigure}[h]{0.45\textwidth}
         \centering
         \includegraphics[width=\textwidth]{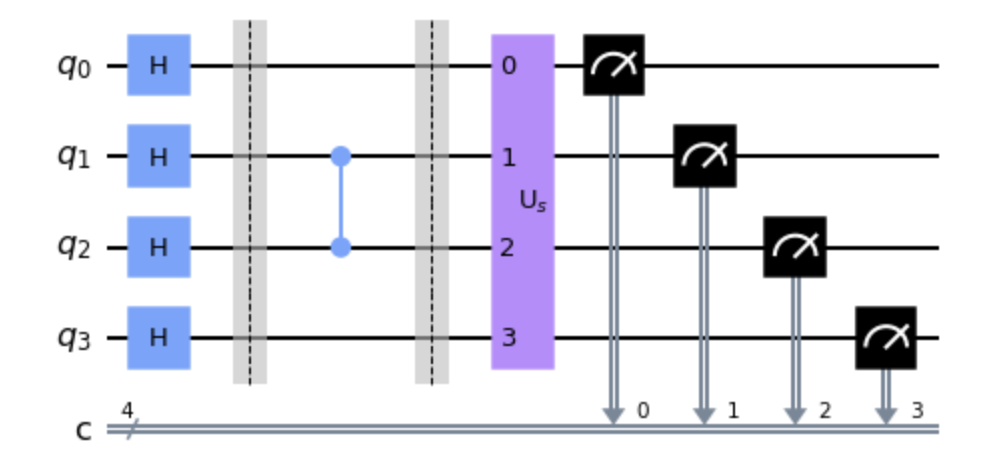}
         \caption{Naive 4-qubit String Detection Circuit}
         \label{fig:invalid_4circuit}
     \end{subfigure}
     \hfill
     \begin{subfigure}[h]{0.45\textwidth}
         \centering
         \includegraphics[width=\textwidth]{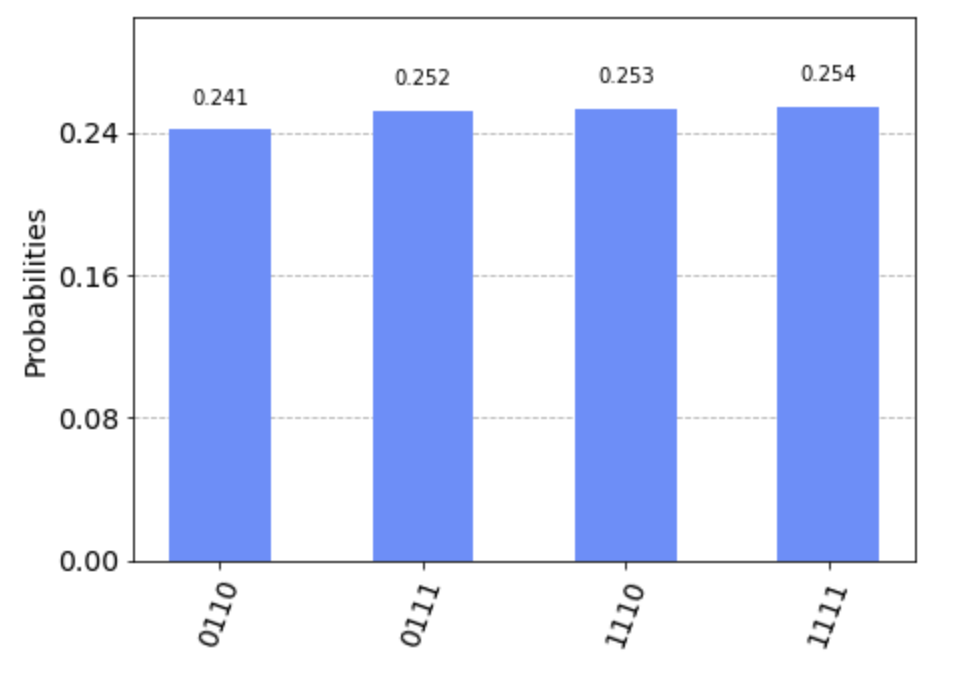}
         \caption{Naive 4-qubit String Detection Result}
         \label{fig:invalid_4result}
     \end{subfigure}
        \caption{4-qubit string detection results with multiple winners. (a) is the circuit of detecting string \textbf{0110}. The oracle is simply a $\mathbi{CZ}$ gate on the middle two qubits. General diffuser stage is applied. (b) is the result, we expect to see $\ket{0110}$ standing out, but instead receive $\sum_{i,j}^2 \frac{1}{2} \ket{i11j}$ each state $\ket{i11j}$ with probability $p$ around $\frac{1}{4}$.}
        \label{fig:invalid_4}
\end{figure}

Therefore, we introduce the multi-controlled multi-target gate ($\mathbi{MCMT}$) \cite{schuch2003programmable} and an ancilla qubit \cite{jqc.2021.018114} in our $4$-qubit circuit, so that the $\mathbi{Z}$ gate will be applied to the ancilla bit if and only if the first $4$ qubits represent the winner. This will help us further collapse the probability onto the winners. Furthermore, based on the work first proved by Bennett et al. \cite{bennett1997strengths} and later enhanced by \cite{boyer1998tight}, for multiple winners in GSA, the same algorithm framework can be applied, but the number of iterations of diffusing stage we need to apply changes from $O(\sqrt{N})$ to $O(\sqrt{N/k})$, where $k$ is the number of winners we have. Fig.\ref{fig:4-results} are the results based on our implementation for $k=2,3,4$ respectively. The circuit of our implementation is shown in Fig.\ref{fig:4circuit} in Appendix.

\begin{figure}[!ht]
     \centering
     \begin{subfigure}[h]{0.9\textwidth}
         \centering
         \includegraphics[width=0.5\textwidth]{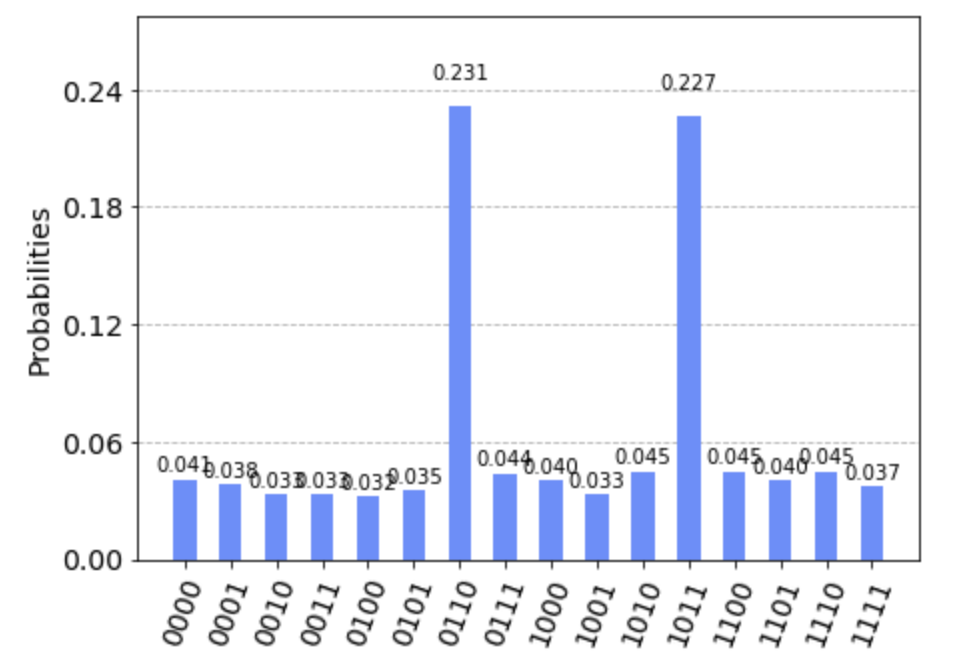}
         \caption{4-qubit with 2 winners \textbf{0110} and \textbf{1101}}
         \label{fig:42}
     \end{subfigure}
     \hfill
     \begin{subfigure}[h]{0.45\textwidth}
         \centering
         \includegraphics[width=\textwidth]{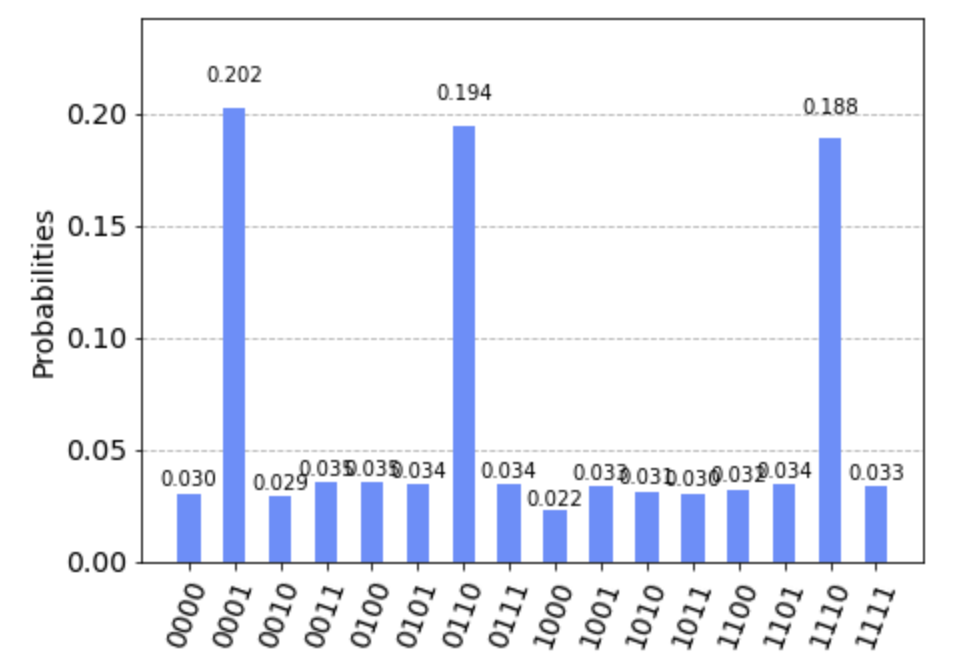}
         \caption{4-qubit with 3 winners \textbf{1000}, \textbf{0110} and \textbf{0111}}
         \label{fig:43}
     \end{subfigure}
     \hfill
     \begin{subfigure}[h]{0.45\textwidth}
         \centering
         \includegraphics[width=\textwidth]{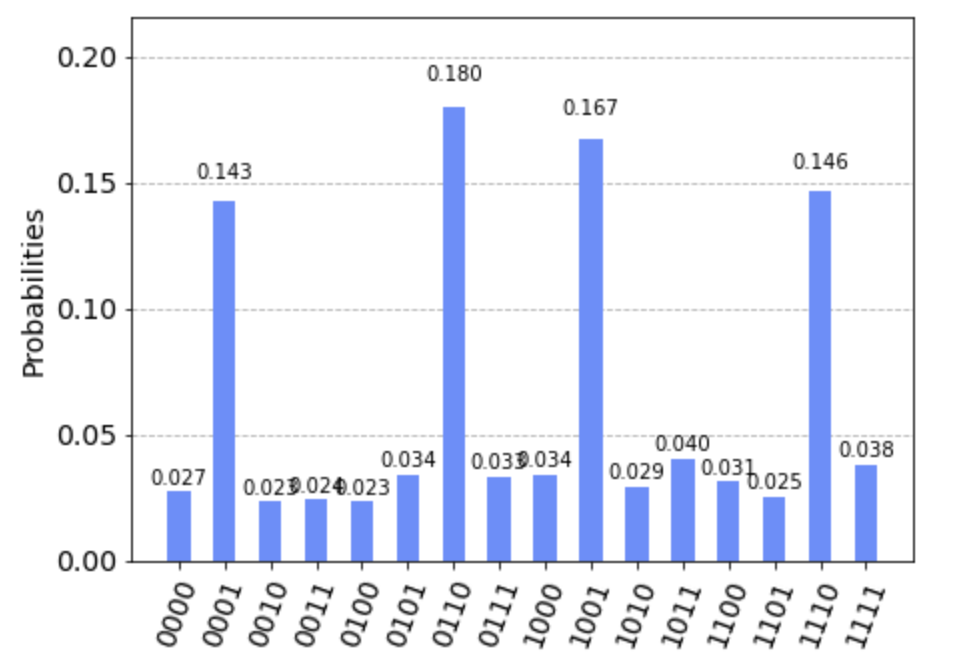}
         \caption{4-qubit with 4 winners \textbf{1000}, \textbf{0110}, \textbf{1001} and \textbf{0111}}
         \label{fig:44}
     \end{subfigure}
        \caption{4-qubit String Detection Results with multiple winners.}
        \label{fig:4-results}
\end{figure}

A simple analysis of the circuit is given for winners \textbf{1101}$\rightarrow \ket{1101}$ and \textbf{0110}$\rightarrow\ket{0110}$:

\begin{enumerate}
\item 
After state preparation on all $4$ qubits, we have them to be on superposition state $\cfrac{1}{4}\sum_{i=0}^{15}\ket{i}$.

\item 
We then query the oracle function for multi winners with $\mathbi{MCMT}$ gate targeting on the ancilla bit. The result of the four testing qubits is:
    \begin{equation}
        \cfrac{1}{4}\sum_{i=0,i\neq 6,13}^{15}\ket{i} - \cfrac{1}{4}\ket{1101} - \cfrac{1}{4}                                           \ket{0110}
    \end{equation}

\item 
Then inside the diffuser stage:
    \begin{enumerate}
        \item Apply $\mathbi{H}$ gates on all $4$ qubits, so that the state will be mapped back to $\ket{0}$:
            \begin{equation}
                \begin{split}
                \ket{\psi} =& \cfrac{3}{4}\ket{0000} -
                \cfrac{1}{4} (\ket{0111} + \ket{1001} + \ket{1110})\\
                &+ \cfrac{1}{4} (\ket{0011} +
                \ket{0100} + \ket{1010} + \ket{1101}) 
                \end{split}
            \end{equation}
        \item Following from this step, we calculate $2\ket{s}\bra{s} - I$ where $\ket{s}$ is the initial state, which flips the states around the mean. We first apply $\mathbi{X}$ gates on all $4$ qubits:
            \begin{equation}
                \begin{split}
                    \ket{\psi} =&\cfrac{3}{4}\ket{1111} - \cfrac{1}{4} (\ket{1000} + \ket{0110} + \ket{0001})\\
                    &+ \cfrac{1}{4} (\ket{1100} + \ket{1011} + \ket{0101}+\ket{0010})
                \end{split}
            \end{equation}
        \item And then apply a $\mathbi{CCCZ}$ gate among $4$ qubits with the last one as the targeting qubit:
            \begin{equation}
                \begin{split}
                    \ket{\psi} =&-\cfrac{3}{4}\ket{1111} - \cfrac{1}{4} (\ket{1000} + \ket{0110} + \ket{0001})\\
                    &+ \cfrac{1}{4} (\ket{1100} + \ket{1011} + \ket{0101}+\ket{0010})
                \end{split}
            \end{equation}
        \item The last step is to apply $\mathbi{X}$ gates again:
            \begin{equation}
                \begin{split}
                    \ket{\psi} =&-\cfrac{3}{4}\ket{0000}- \cfrac{1}{4} (\ket{0111} + \ket{1001} + \ket{1110})\\
                &+ \cfrac{1}{4} (\ket{0011} +
                \ket{0100} + \ket{1010} + \ket{1101}) 
                \end{split}
            \end{equation}
        \item Finally, we add $\mathbi{H}$ gates to all 4 qubits to revert them back from $\ket{0}$ to $\ket{s}$:
            \begin{equation}
                \begin{split}
                    \ket{\psi} =&-\cfrac{5}{8}(\ket{0110} + \ket{1101}) - \cfrac{1}{8} \sum_{i=0, i \neq 6, 13}^{15} \ket{i}
                \end{split}
            \end{equation}
    \end{enumerate}
\end{enumerate}

Since we have $16$ vectors and $2$ winners, roughly $2$ iterations should be enough. We perform only $1$ iteration on the simulator and the result is desired. The amplitude for winner states $\ket{0110}$ and $\ket{1101}$ is theoretically around $0.39$, and the final measurement we have is around $0.23$, which is already considerably larger than the other states.

\subsection{Orthogonal Vector Problem}
\subsubsection{GSA}\label{gsa_case}
Unlike the SDP, OVP does not have a static oracle. Thus we introduce a reference qubit(s) representing one of the vectors inside the final orthogonal vector pairs that we might detect. Moreover, the oracle needs to indicate whether two state vectors are orthogonal to each other. This can be achieved by embedding a controlled-SWAP ($\mathbi{c-SWAP}$) gate \cite{buhrman2001quantum} between a testing qubit and the reference qubit:
\begin{equation}
    \mathbi{c-SWAP}_{q,\ket{\phi}\ket{\psi}} = \ket{0}\bra{0}\otimes I\otimes I+\ket{1}\bra{1}\otimes \mathbi{SWAP} =
\begin{pmatrix}
    1 &  0 &  0 &  0 &  0 &  0 &  0 &  0\\
    0 &  1 &  0 &  0 &  0 &  0 &  0 &  0\\
    0 &  0 &  1 &  0 &  0 &  0 &  0 &  0\\
    0 &  0 &  0 &  1 &  0 &  0 &  0 &  0\\ 
    0 &  0 &  0 &  0 &  1 &  0 &  0 &  0\\
    0 &  0 &  0 &  0 &  0 &  0 &  1 &  0\\
    0 &  0 &  0 &  0 &  0 &  1 &  0 &  0\\
    0 &  0 &  0 &  0 &  0 &  0 &  0 &  1
\end{pmatrix}
\end{equation}

Given three qubits $\ket{\phi}$ $\ket{\psi}$ and a controlling qubit $q$, if $q =\ket{1}$, the gate changes $\ket{\phi}\ket{\psi}$ to $\ket{\psi}\ket{\phi}$; if $q = \ket{0}$, the $\mathbi{c-SWAP}$ gate does nothing. Therefore, the final state of $\mathbi{c-SWAP}$ gate given $\ket{\phi}$, $\ket{\psi}$ can be written as:
\begin{align}
    \cfrac{1}{2}\ket{0}\left(\ket{\phi}\ket{\psi}+\ket{\psi}\ket{\phi}\right)+\cfrac{1}{2}\ket{1}\left(\ket{\phi}\ket{\psi}-\ket{\psi}\ket{\phi}\right)
\end{align}

For orthogonal qubits $\ket{\phi}, \ket{\psi}$ such that $\bra{\psi} \phi \rangle = 0$, the probability of the controlling qubit to be measured as $\ket{1}$ is $ \cfrac{1}{2}$. For nearly identical $\ket{\phi}$ and $\ket{\psi}$, the probability will instead be 0, which means that the measurement will always be pure $\ket{0}$. Therefore, we aim to detect an orthogonal pair by detecting a mixed state $\frac{1}{2}\left(\ket{0} + \ket{1}\right)$. Notice that in order to embed $\mathbi{c-SWAP}$, we  need to prepare a ``control" qubit for each pair of qubits performing the swap operation.

Under the GSA setting, we only consider the trivial case such that all vectors are in $\{0,1\}^2$, i.e., one qubit will represent one single two-dimensional vector. Therefore, there will be $2n+1$ qubits, with $n$ qubits representing $n$ vectors $\in \{0,1\}^2$ and another $n$ qubits as the measuring qubits for each vector qubits to perform the SWAP test. We only apply the oracle and the diffuser stage to the measuring qubits instead of the vector qubits we are testing against. The result circuit is in Fig.\ref{fig:ovp-circuit} under Appendix. Each bit $r_i$ in the final measurement can be viewed as an indicator of whether the testing vector represented by $q_i$ is orthogonal to the reference vector. 

We prepare the reference qubit to represent vector $\left(\begin{array}{cc}
    0 & 1 
\end{array}\right)^T$, and two testing qubits to be on $\left(\begin{array}{cc}
    0 & 1 
\end{array}\right)^T$ and $\left(\begin{array}{cc}
    1 & 0 
\end{array}\right)^T$ respectively. Thus we should see the state $\ket{10}$ with top measurement probability at the end, i.e., the measurement should clearly indicate that the second testing qubit stores a vector orthogonal to the reference vector. However, the result as shown in Fig.\ref{fig:ovp-GSA-result} is not desired. Instead of a single winner $\ket{10}$, we have both $\ket{10}$ and $\ket{00}$.

\begin{figure}[h]
    \centering
    \includegraphics[width=0.4\textwidth]{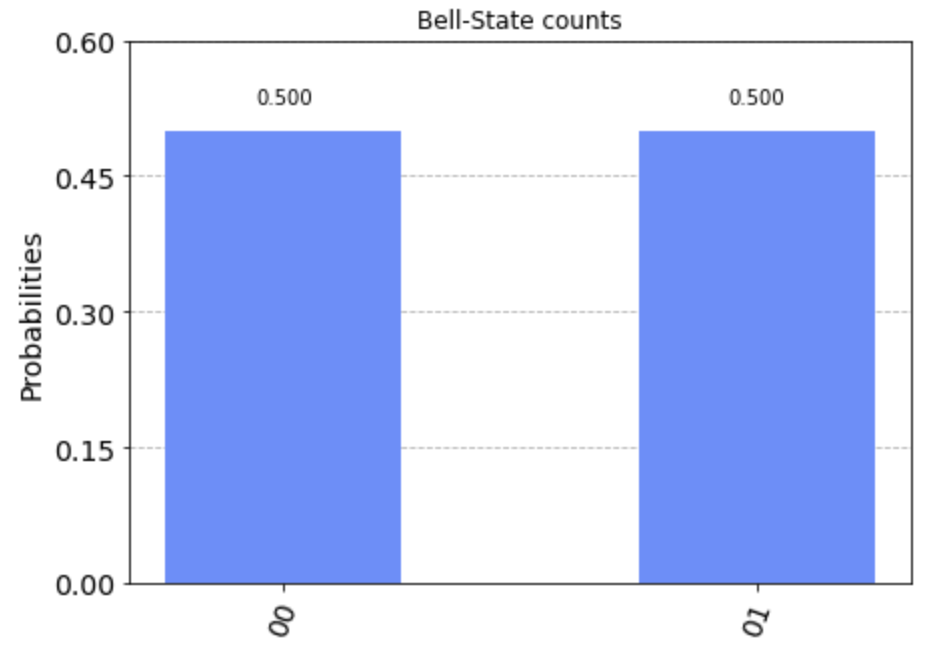}
    \caption{Measurement Results for OVP under GSA setting.}
    \label{fig:ovp-GSA-result}
\end{figure}

If we analyze the circuit by skipping the state preparation and the SWAP test but directly setting the second qubit as $\ket{+} = \frac{1}{\sqrt{2}} \left(\ket{0} + \ket{1}\right)$, we will obtain the following derivation:

\begin{enumerate}
    \item After the SWAP test, the first two qubits are on the state:
        \begin{equation}
            \cfrac{1}{\sqrt{2}}\left(\ket{00} + \ket{10}\right)
        \end{equation}
        
    \item Apply the oracle to both, which flips everything except $\ket{00}$:
        \begin{equation}
            \cfrac{1}{\sqrt{2}}\left(\ket{00} - \ket{10}\right)
        \end{equation}
        
    \item Apply $\mathbi{H}^{\otimes 2}$ gates to both: 
        \begin{equation}
            \cfrac{1}{2\sqrt{2}}\left(\ket{00} + \ket{10} +  \ket{01} + \ket{11} + \ket{11} - \ket{00} + \ket{10} - \ket{01})\right) = \frac{1}{\sqrt{2}}\left(\ket{10}+\ket{11}\right)
        \end{equation}
        
    \item Apply $\mathbi{CZ}\times \mathbi{Z}^{\otimes 2}$ gates to both, which flips all states except $\ket{00}$, we then have
        \begin{equation}
            \frac{1}{\sqrt{2}}\left(-\ket{10}-\ket{11}\right)
        \end{equation}
        
    \item Apply $\mathbi{H}^{\otimes 2}$ gates again to both: 
        \begin{equation}
            \cfrac{1}{\sqrt{2}}\left(-\ket{00}+\ket{10}\right)
        \end{equation}
        
\end{enumerate}

Therefore, if the oracle outputs two measuring qubits on a mixed state $\frac{1}{\sqrt{2}} \left(\ket{00} - \ket{10}\right)$ after applying the diffuser, they would end up with the state $\frac{1}{\sqrt{2}} \left(-\ket{00} + \ket{10}\right)$. Only the phase for $\ket{10}$ is flipped from $\pi$ back to $0$, while the amplitude remains the same. This indicates that with non-uniform distribution, the general GSA algorithm does not work as expected.

Similar case happens when we apply the amplitude amplification on each single measuring qubit. If the qubit represents a testing vector that is orthogonal to the reference vector, then it will be on the state $\frac{1}{\sqrt{2}}\left(\ket{0} +\ket{1}\right)$ which is uniform, while in the other case, it will be just $\ket{0}$. We want to find a gate $U_f = \left(\begin{array}{cc}
    a & b \\
    c & d
\end{array}\right)$ such that: $U_f\ket{+} = \ket{1}$ and $U_f\ket{0} = \ket{0}$ by solving the following equation system:
\[
\left(\begin{array}{cc}
    a & b \\
    c & d
\end{array}\right)\left(\begin{array}{c}
    1  \\
    1 
\end{array}\right) = \left(\begin{array}{c}
    0 \\
    1
\end{array}\right),
\left(\begin{array}{cc}
    a & b \\
    c & d
\end{array}\right)\left(\begin{array}{c}
    1  \\
    0 
\end{array}\right) = \left(\begin{array}{c}
    1 \\
    0
\end{array}\right)
\]
 
The result for the equation is $U_f = \left(\begin{array}{cc}
    1 & -1 \\
    0 & 1
\end{array}\right)$, which is clearly not unitary. Therefore, we are not able to decompose it under Pauli base with density matrices, which implies that there are no combinations of quantum logic gates to represent this transformation.

\subsubsection{Modified GSA}\label{sec:modified_gsa}
Biron et al \cite{biron1999generalized} proposed a modified GSA algorithm adapted for arbitrary initial distribution, and the main modification is to omit the first step that puts every qubits in mixed states in the original algorithm and then proceeds with the normal GSA loop that flips the phase of the marked (desired) states and rotates around average of all states. We take our example in \cref{gsa_case} and use $(a_1, a_2, a_3, a_4)$ to represent the amplitude of quantum states $\ket{00}$, $\ket{01}$, $\ket{10}$, and $\ket{11}$ respectively:
\begin{enumerate}
    \item 
    \textbf{Initialization}: Assume the SWAP test was performed before, and we have $\frac{1}{\sqrt{2}}(1,0,\mathbf{1},0)$ representing $\frac{1}{\sqrt{2}}(\ket{00} + \ket{10})$ as the non-uniform initial distribution.
    
    \item
    \textbf{Iteration 1}: (Flip the phase of the marked state and rotate around the average of all states, same below):$\frac{1}{\sqrt{2}}(-1,0,\mathbf{1},0)$
    
    \item
    \textbf{Iteration 2}: $\frac{1}{\sqrt{2}}(0,-1,\mathbf{0},-1)$
    
    \item
    \textbf{Iteration 3}: $\frac{1}{\sqrt{2}}(-1,0,\mathbf{-1},0)$
    
    \item
    \textbf{Iteration 4}: $\frac{1}{\sqrt{2}}(1,0,\mathbf{-1},0)$
    
    \item
    \textbf{Iteration 5}: $\frac{1}{\sqrt{2}}(0,1,\mathbf{0},1)$
    
    \item
    \textbf{Iteration 6}: $\frac{1}{\sqrt{2}}(1,0,\mathbf{1},0)$, which transforms back to our initial distribution.
\end{enumerate}

The above iterations follow exactly the procedure described in \cite{biron1999generalized}. To end the loop, we perform more iterations than required. However, the amplitude of the marked (desired) state $\ket{10}$ (marked in bold) is either $\pm 1$ or $0$, which means that the probability of detecting it is either $1/2$ or totally $0$. In fact, even with the modification proposed in \cite{ventura2000quantum}, which not only flips the phase for marked states, but also for all states that occurs in the original superposition (in our case, they are state $\ket{00}$ and $\ket{10}$) after the first iteration, it still simply transforms $\frac{1}{\sqrt{2}}(-1,0,1,0)$ to $\frac{1}{\sqrt{2}}(1,0,-1,0)$ and then to $\frac{1}{\sqrt{2}}(-1,0,1,0)$ back again. Obviously this does not work well in our example, and this case can be further generalized. We first formalize the framework definition in \cite{biron1999generalized}:

\begin{definition}[Modified GSA]\label{modified_gsa}
For $n$ qubits, denote the desired states $\ket{i}$ to be the marked states. W.L.O.G., we require $|\{\ket{i}\}| \leq N/2, N =2^n$. The rest quantum states are the unmarked states. Denote $k_i(t)$ to be the amplitude of the marked state $\ket{i}$, and $l_j(t)$ to be the amplitude of the unmarked state $\ket{j}$. $k_i(0), l_j(0)$ are the initial amplitudes derived from the initial distribution. Denote the total number of iterations as $T$ where $T\geq 2$. The Modified GSA is as follows:
\begin{itemize}
    \item
    Initialize marked and unmarked states on $n$ qubits to an arbitrary initial distribution.
    
    \item
    For $t \in \{0,\dots, T-1\}$:
    \begin{itemize}
        \item 
        $k_i(t+1) = \frac{2}{N}\left(\sum_j l_j(t) - \sum_i k_i(t)\right) + k_i(t)$
        \item
        $l_j(t+1) = \frac{2}{N}\left(\sum_j l_j(t) - \sum_i k_i(t)\right) - l_j(t)$
    \end{itemize}
\end{itemize}
\end{definition}

\begin{theorem}
For $n$ qubits with $N = 2^n$ possible states, denote the set of all marked state $\ket{i}$ as $W$ and denote $S$ to be the set of all states in the initial distribution that has non-zero amplitude. If $|W| = |S\backslash W| = \frac{N}{4}$ and $W \subset S$, i.e., we have in total $\frac{N}{2}$ states that have non-zero amplitudes in the initial distribution, while half of them form the whole set of the marked states. If all of the states in $S$ are on the same amplitude initially, then the probability of measuring a specific marked state based on the algorithm defined in \cref{modified_gsa} is either $0$ or $\frac{2}{N}$ for $T \geq 2$.
\end{theorem}

\begin{proof}
Since all states in $S$ have the same amplitude in the initial distribution, all marked states $\ket{i}$ act exactly the same in \cref{modified_gsa}. Thus, w.l.o.g., we can merge them into one state with higher weight (i.e., amplitude $=\frac{1}{\sqrt{2}}$). Similarly, states in $S \backslash W$ can also be merged into one state with amplitude $\frac{1}{\sqrt{2}}$. For states with zero amplitude in the initial distribution, although they contribute equally in the amplitude average, to make things symmetric, they can be merged into two states each with an initial amplitude $0$. Now, w.l.o.g, we assume the initial distribution is $\frac{1}{\sqrt{2}}(1,0,1,0)$, i.e., state $\ket{10}$ is the merged marked state for $W$, state $\ket{00}$ is the merged states for states in $S \backslash W$, and states $\ket{01}$ and $\ket{11}$ are the merged states for the other states with zero initial amplitude. The analysis follows exactly as above and the probability of state $\ket{10}$ (again, w.l.o.g) is either $0$ or $(\pm 1/\sqrt{2})^2 = \frac{1}{2}$. Notice that the property that all marked states sharing the same amplitude is also preserved during the procedure defined in \cref{modified_gsa}. Therefore, for a specific marked state, the probability of detecting it in the final measurement is either $0$ or $\frac{1/2}{N/4} = \frac{2}{N}$.
\end{proof}

\subsubsection{Extended Version}\label{extend}
If we step back from GSA, but only look at the OVP, we notice that if we consider the OVP for binary vectors only, it seems straight forward to calculate the orthogonal vector by simply outputting the one's complement of the reference qubit(s). Hence, we define a new quantum setting of OVP that works for vectors of arbitrary dimension as following:

\begin{definition}[Quantum One's Complement]\label{def:one_comp}
Given a $n$-bit binary vector $\in \{0,1\}^n$, initialize $n$ reference qubits with each qubit representing 1 bit of the given vector and $n$ testing qubits on mixing states, find a quantum circuit $\mathbi{QC}$ such that the final measurements $\ket{i}_{i\in [2^n-1]}$ on $n$ testing qubits is a quantum-representation of the reference vector's one's complement.
\end{definition}

For example, if the reference qubit(s) are $\ket{101}$, the final measurement on testing qubits should be $\ket{010}$. Notice that for any vector $v \in \{0,1\}^n$, $\langle v, \mathbf{0} \rangle  = 0$. Hence, to simplify the algorithm, we require that for an orthogonal vector pair $(v, w)$, $v_i \neq w_i$.

The state preparation is the same as the one for GSA that we put $n$ testing qubits into a mixed state, representing all possible vectors $\in \{0,1\}^n$. If a reference qubit is $\ket{1}$, we want to output $\ket{0}$, and $\ket{1}$ otherwise. Hence a we apply a $\mathbi{CZ}$ gate followed by another $\mathbi{H}$ gate. No diffuser or any other kind of amplitude amplification is needed here since the result is deterministic. The implemented circuit is shown as Fig.\ref{fig:ovp-no-GSA} under Appendix.

Following is a case analysis for the reference vector $\left(\begin{array}{ccc}
     1 & 0 & 1
\end{array}\right)^T$:

\begin{enumerate}
    \item Given reference qubits under state $\ket{101}$, after state preparation, we have $\ket{010}$ on reference qubits, and testing qubits are under:
        \begin{equation}
            \cfrac{1}{2\sqrt{2}}\left(\ket{000} + \ket{001} + \ket{010} + \ket{011} + \ket{100} + \ket{101} + \ket{110} + \ket{111}\right)
        \end{equation}
        
    \item After applying $\mathbi{CZ}$ gates to control qubits $q_3$, $q_4$, $q_5$ with state $\ket{010}$ and targeting qubits $q_0$, $q_1$, $q_2$, the first three qubits become:
        \begin{equation}
            \begin{split}
                 \cfrac{1}{2\sqrt{2}}\left((\ket{0} + \ket{1})\otimes (\ket{0} - \ket{1})\otimes (\ket{0} + \ket{1})\right)
            \end{split}
        \end{equation}
        
    \item Then we apply $\mathbi{H}$ gates again on the first three qubits:
        \begin{equation}
             \ket{0}\otimes\ket{1}\otimes\ket{0} = \ket{010}
        \end{equation}
\end{enumerate}

Therefore, in this case we can output the a fresh copy of the one's complement with probability $=1$. Since all $\mathbi{CZ}$ gates can be applied simultaneously to each pair of qubits (one as a reference qubit, one as a testing qubit), the depth of the given circuit is constant, while under the classical setting, we need to traverse all bits one by one. Moreover, unlike simply flipping the qubits from $\ket{0}$ to $\ket{1}$ or vice versa, by applying $\mathbi{CZ}$ gates on a reference qubit and $k \geq 1$ testing qubits, we are able to extend the circuit and make $k$ copies of the desired one's complement for some reasonable chosen $k$, thus bypass the no-cloning problem.

We want to point out that finding the one's complement of a $n$-bit vector in classical world can also be considered as a simple task with a naive algorithm run in $O(d)$. However, our short-depth circuit is able to identify the orthogonal vector in constant time. Besides, the same circuit can be easily extended to cases where the reference qubits are not in pure states. As shown in Fig.\ref{fig:mixed}, if the reference vector(s) is either $\left(\begin{array}{ccc}
    1 & 0 & 1
\end{array}\right)^T$ or $\left(\begin{array}{ccc}
    0 & 0 & 1
\end{array}\right)^T$, we can simply apply the Hadamard gate on the first reference qubit, and perform the same algorithm. The measurement will be $\ket{010}$ and $\ket{110}$ with almost even distribution, which is as expected. 

\begin{figure}[ht]
     \centering
     \begin{subfigure}[h]{0.45\textwidth}
         \centering
         \includegraphics[width=\textwidth]{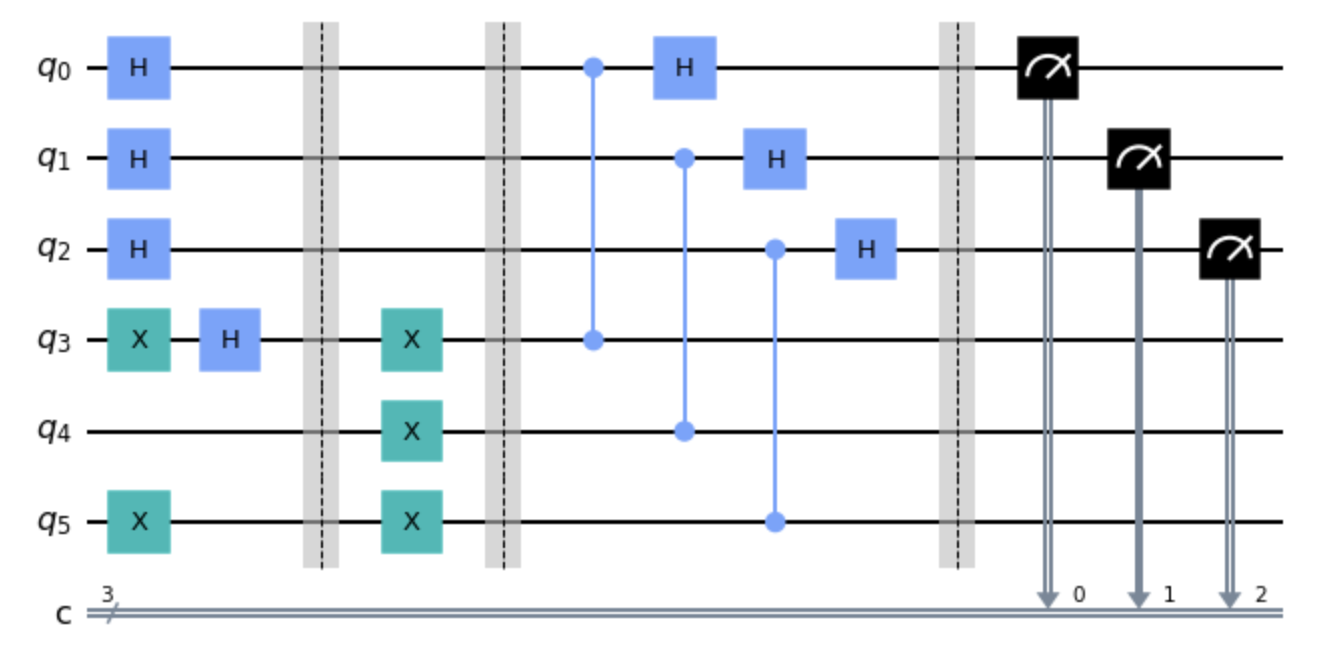}
         \caption{Circuit for OVP with mixed reference qubits.}
         \label{fig:mix-cir}
     \end{subfigure}
     \hfill
     \begin{subfigure}[h]{0.45\textwidth}
         \centering
         \includegraphics[width=\textwidth]{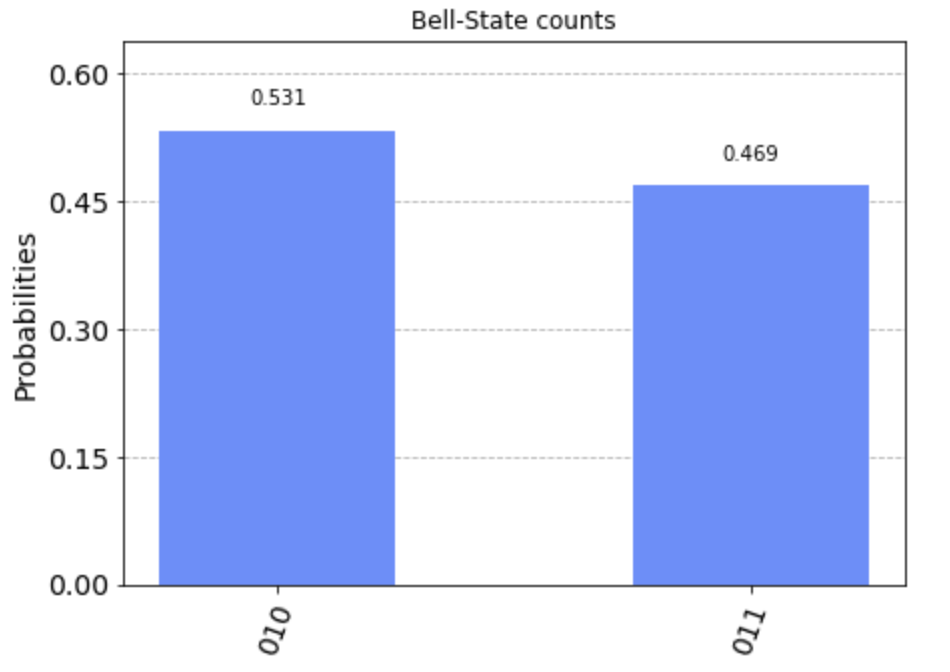}
         \caption{Result for OVP with mixed reference qubits.}
         \label{fig:mix-res}
     \end{subfigure}
        \caption{Simulation for OVP with mixed reference qubits.}
        \label{fig:mixed}
\end{figure}

\section{Conclusion}\label{conclusion}
Inspired on GSA, we first extend the results of the SDP under a 4-qubit quantum circuit with multiple winners. Then we give out two definitions of OVP under quantum settings. In \cref{def:ovp_1}, we make use of the GSA framework and treat the orthogonality test as an oracle with a general diffuser stage followed. We further investigate the behavior of a special initial distribution under modified GSA proposed by Biron et al \cite{biron1999generalized}. Unfortunately, the result does not collapse to the marked states but remains mixed. Under \cref{def:one_comp}, we focus on the OVP itself and utilize the idea of one's complement to implement a constant short-depth circuit with a deterministic result. 

Two main contributions are made in this work. First, based on the case analysis we conduct in \cref{sec:modified_gsa}, we propose a theorem that generalize the condition under which even the modified GSA framework does not apply. Moreover, we argue that treating each qubit as a $2$-dimensional vector as in \cref{def:ovp_1} may not be meaningful. The foremost observation is that in order to perform the SWAP test on all testing qubits neutrally, we need to hold $n$ copies of the reference qubits. Whereas it is impossible to clone an identical and independent copy of the reference qubit without disturbing it \cite{wootters1982single}. Additionally, we need $n$ measuring qubits for performing $n$ SWAP tests. Due to these limitations, we claim that we should focus on \cref{def:one_comp} based on one's complement. With the power of quantum computing, OVP of our definition can be solved in constant time with a constant-depth quantum circuit compared to $O(\sqrt{N/k})$ for GSA under the standard settings. Further research may be conducted on OVP and its variations to see if our implementation could be potentially useful for larger circuits.

\pagebreak
\addcontentsline{toc}{section}{References}
\bibliographystyle{alpha}
\bibliography{ref}

\pagebreak 
\appendix

\section{Circuits}
\begin{figure}[h]
    \centering
    \includegraphics[width=\textwidth]{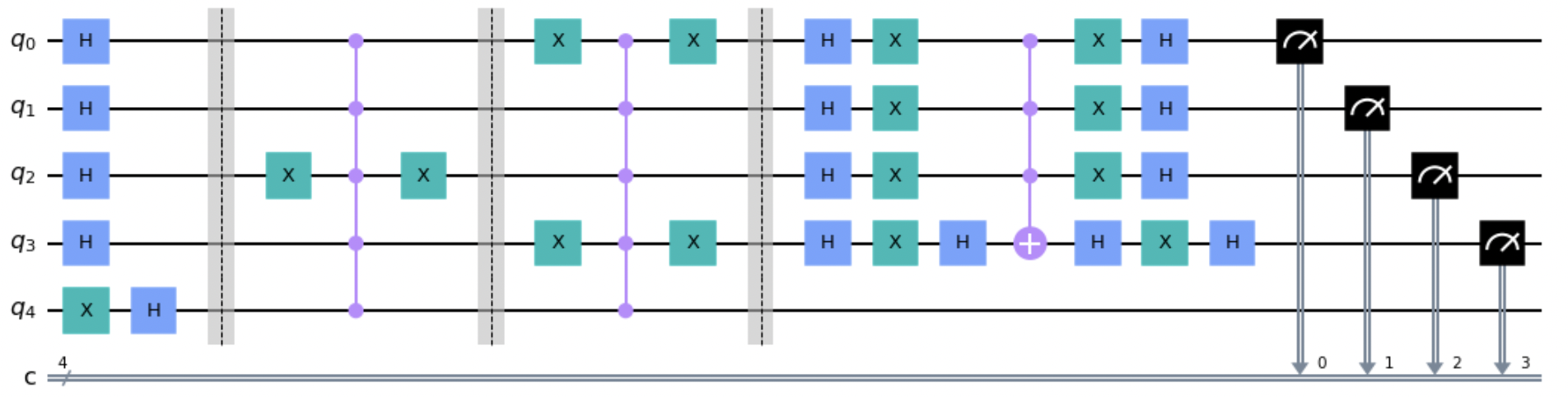}
    \caption{4-qubit String Detection Circuit with multiple winners. $q_4$ is the ancilla qubit. The winner for this circuit is $\mathbf{1101}$ and $\mathbf{0110}$. One iteration of general diffuser is applied.}
    \label{fig:4circuit}
\end{figure}

\begin{figure}[h]
    \centering
    \includegraphics[width=\textwidth]{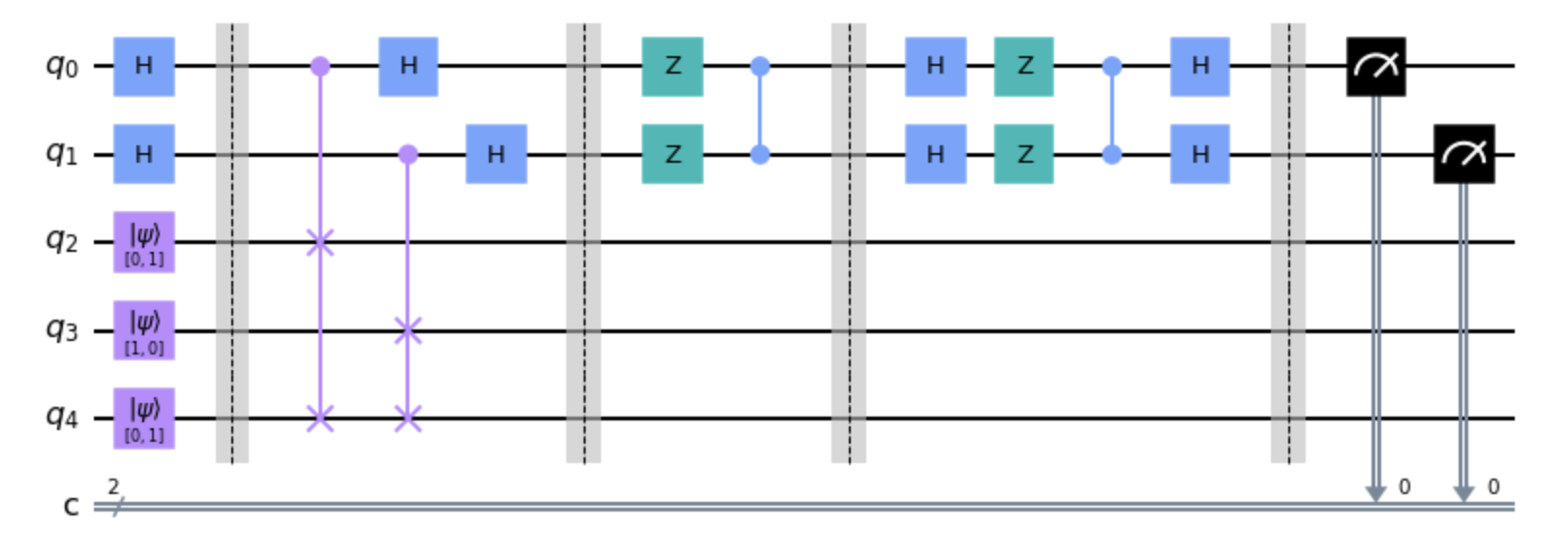}
    \caption{Orthogonal Vector Problem Circuit under GSA. $q_0$ and $q_1$ are the measuring qubits. $q_2$ and $q_3$ are the testing qubits representing different vectors $\in \{0,1\}^2$. $q_4$ is the reference qubit, representing one of the vectors in the final orthogonal vector pairs that we might detect. The SWAP test is first performed between testing qubits and the reference qubit. The oracle flips all state except for $\ket{00}$. One iteration of general $2$-qubit diffuser is applied.}
    \label{fig:ovp-circuit}
\end{figure}

\begin{figure}[ht]
    \centering
    \includegraphics[width=0.8\textwidth]{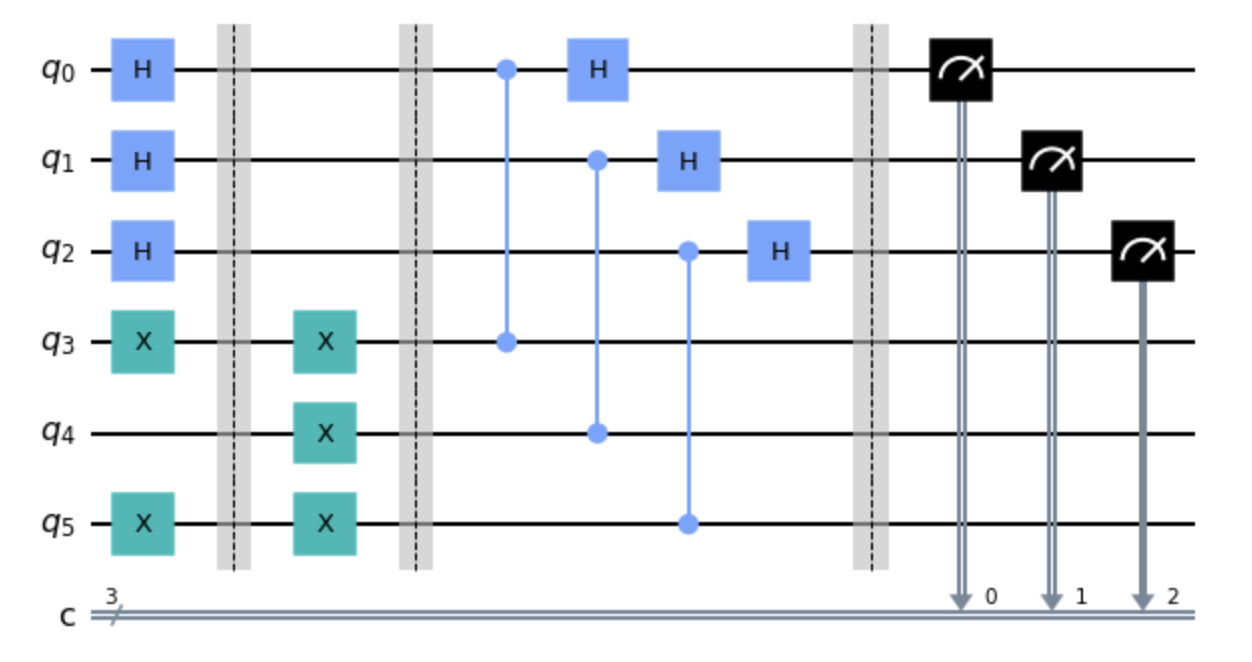}
    \caption{OVP Circuit without applying GSA. $q_3, q_4, q_5$ are the reference qubits, which together representing vector $\left(\begin{array}{ccc}
        1 & 0 & 1  
    \end{array}\right)^T$. $q_0, q_1, q_2$ represent all possible vectors $\in \{0,1\}^3$. The expected result should be $\left(\begin{array}{ccc}
        0 & 1 & 0  
    \end{array}\right)^T$. The logic is to flip the state from $\ket{0} + \ket{1}$ to $\ket{0} - \ket{1}$ if the control qubit in the reference vector is $\ket{0}$, after applying Hadmard gate the measuring qubit will be mapped to pure $\ket{1}$ state while other qubits will be mapped to pure $\ket{0}$ state. The whole process is similar to applying GSA without the diffusing stage.}
    \label{fig:ovp-no-GSA}
\end{figure}

\end{document}